\documentclass[sigconf]{acmart}
\setcopyright{acmcopyright}
\copyrightyear{2025}
\acmYear{2025}

\acmConference[]{}{}{}
\acmBooktitle{}
\acmPrice{}
\acmISBN{}
\acmDOI{}

\usepackage{amsthm,amsmath}
\usepackage{tikz}

\newtheorem{theorem}{Theorem}

\newtheorem{thm}[theorem]{Theorem}

\newtheorem{lemma}[theorem]{Lemma}

\newtheorem{defi}[theorem]{Definition}

\def\eatspace#1{#1}
\def\step#1#2{\par\kern1pt\hangindent#2em\hangafter=1\noindent\rlap{\small#1}\kern#2em\relax\eatspace}

\let\set\mathbb
\def\<#1>{\langle#1\rangle}

\begin{document}
\fancyhead{}
\title{Flip Graphs for Polynomial Multiplication}
\thanks{S.\ Chen was partially supported by the National Key R\&D Program of China (No. 2023YFA1009401),
the NSFC grant (No. 12271511), CAS Project for Young Scientists in Basic Research (Grant
No. YSBR-034), and the CAS Fund of the Youth Innovation Promotion Association (No. Y2022001). 
M.\ Kauers was supported by the Austrian FWF grants 10.55776/PAT8258123, 10.55776/I6130, and 10.55776/PAT9952223.
}

\author{Shaoshi Chen}
\affiliation{%
  \institution{KLMM, Academy of Mathematics and Systems Science}
  \institution{Chinese Academy of Sciences}
  \institution{School of Mathematical Sciences}
  \institution{University of Chinese Academy of Sciences}
  \state{}
  \postcode{Beijing 100190}
  \country{China}
}
\email{schen@amss.ac.cn}

\author{Manuel Kauers}
\affiliation{%
  \institution{Institute for Algebra}
  \institution{Johannes Kepler University}
  \state{}
  \postcode{4040 Linz}
  \country{Austria}
}

\email{manuel.kauers@jku.at}

\begin{abstract}
  Flip graphs were recently introduced in order to discover new matrix
  multiplication methods for matrix sizes. The technique applies to
  other tensors as well. In this paper, we explore how it performs for
  polynomial multiplication.
\end{abstract}


\maketitle

\section{Introduction}

Matrix multiplication can be viewed as a bilinear map
\[
  K^{n\times k}\times K^{k\times m}\to K^{n\times m},
  \quad(A,B)\mapsto AB.
\]
As such, it can also be described as an element of the
tensor product space $K^{n\times k}\otimes K^{k\times m}\otimes K^{n\times m}$.
Indeed, if we write $a_{i,j}$ for the element of $K^{n\times k}$
having a $1$ at position $(i,j)$ and zeros in all other positions,
and we write $b_{i,j}$ and $c_{i,j}$ for the elements of $K^{k\times m}$
and $K^{n\times m}$ defined analogously, then the matrix multiplication
tensor can be written as
\[
  \sum_{u=1}^n\sum_{v=1}^k\sum_{w=1}^m a_{u,v}\otimes b_{v,w}\otimes c_{u,w}.
\]
This is a sum of $nkm$ ``pure'' tensors, i.e., tensors that can be written
in the form $M\otimes M'\otimes M''$ for certain $M\in K^{n\times k}$, $M'\in K^{k\times m}$
and $M''\in K^{n\times k}$.

There are other, less obvious ways to write the matrix multiplication tensor
as a sum of pure tensors. The \emph{rank} of a tensor $T$ is defined as the
minimal number of pure tensors such that $T$ is the sum of these tensors.
Pure tensors are therefore tensors of rank~$1$.

The quest for fast matrix multiplication algorithms boils down to the question
what the rank of the matrix multiplication tensor is. We do not know. For $n,k,m$
all equal and asymptotically large, the rank is $O(n^\omega)$ for a certain
number $\omega$ which is known to be less than 2.371~\cite{williams24,alman24}. For various
small and specific values of $n,k,m$, Sedoglavic~\cite{fastmm} keeps track of the
best known upper bounds for the corresponding ranks.

Low rank tensor decompositions for small and specific matrix formats can be
found in several ways. Some have been found by hand~\cite{St:Gein,La:Anaf,HK:OMtN}, some by numerical
techniques~\cite{Sm:Tbca,SS:TTRo}, some by SAT solvers~\cite{HKS:Nwtm,heule19a,CBH:ANGM}, some by machine
learning~\cite{FBH+:Dfmm}. Most recently, Kauers and Moosbauer~\cite{kauers23f} proposed
the idea of flip graphs for searching for such schemes. The idea of flip
graphs is summarized in Section~\ref{sec:flipgraphs} below. 

These search techniques are not restricted to the matrix multiplication tensor
but can also be applied to other bilinear maps. For example, if $K[x]_{\leq n}$
denotes the vector space of all polynomials of degree at most~$n$, then polynomial
multiplication can be viewed as a bilinear map
\[
K[x]_{\leq n}\times K[x]_{\leq m}\to K[x]_{\leq n+m}
\]
and thus as an element of the tensor product
\[
K[x]_{\leq n}\otimes K[x]_{\leq m}\otimes K[x]_{\leq n+m}.
\]
Stated in this language, the multiplication tensor for $n=m=1$ reads
\begin{alignat*}1
     &a_0\otimes b_0\otimes c_0\\
  {}+{}&a_1\otimes b_0\otimes c_1\\
  {}+{}&a_0\otimes b_1\otimes c_1\\
  {}+{}&a_1\otimes b_1\otimes c_2.
\end{alignat*}
By adding the term $a_0\otimes b_0\otimes c_1$ to the second summand
and subtracting it from the first, and by adding the term $a_1\otimes b_1\otimes c_1$
to the third summand and subtracting it from the fourth, we can see that
the tensor can also be written as follows:
\begin{alignat*}1
     &a_0\otimes b_0\otimes (c_0-c_1)\\
  {}+{}&(a_0+a_1)\otimes b_0\otimes c_1\\
  {}+{}&(a_0+a_1)\otimes b_1\otimes c_1\\
  {}+{}&a_1\otimes b_1\otimes (c_2-c_1).
\end{alignat*}
Now the second and the third summand can be merged into one, and we obtain
\begin{alignat*}1
     &a_0\otimes b_0\otimes (c_0-c_1)\\
  {}+{}&(a_0+a_1)\otimes (b_0+b_1)\otimes c_1\\
  {}+{}&a_1\otimes b_1\otimes (c_2-c_1).
\end{alignat*}
This is exactly Karatsuba's algorithm in tensor notation, derived with the
idea of flip graphs.

The purpose of this paper is to explore more generally how the flip graph technique
performs for the polynomial multiplication tensor. Polynomial multiplication is
much better understood than matrix multiplication. In particular, we know that
the tensor rank for multiplying polynomials of degree $n$ with polynomials of
degree $m$ (over a sufficiently large coefficient field) is equal to
$n+m-1$. Our goal is not to use flip graphs to learn something new about
polynomial multiplication, but to use polynomial multiplication to learn
something new about flip graphs. Our main result is that the flip graph
technique can find the best possible multiplication algorithms (again assuming a
sufficiently large coefficient field), and to bound the length of the path from
the standard algorithm to the optimal algorithm. In addition, in Section~\ref{sec:chartwo}
we report on some experiments we did searching for polynomial multiplication
algorithms for the coefficient field~$\set Z_2$.

\section{Tensors and Flip Graphs}\label{sec:flipgraphs}

Recall that the tensor product $U\otimes V\otimes W$ of three vector spaces
$U,V,W$ over some field $K$ consists of all $K$-linear combinations of
equivalence classes of elements of $U\times V\times W$ subject to the relations
\begin{alignat*}1
  (u_1+u_2,v,w) &= (u_1,v,w)+(u_2,v,w)\\
  (u,v_1+v_2,w) &= (u,v_1,w)+(u,v_2,w)\\
  (u,v,w_1+w_2) &= (u,v,w_1)+(u,v,w_2)\\
  \alpha(u,v,w) &= (\alpha u,v,w)=(u,\alpha v,w)=(u,v,\alpha w)
\end{alignat*}
for all $\alpha\in K$, $u,u_1,u_2\in U$, $v,v_1,v_2\in V$, and $w,w_1,w_2\in W$.
Instead of $(u,v,w)$, we write $u\otimes v\otimes w$. Every element of the
tensor product $U\otimes V\otimes W$ can thus be written in the form
\[
  (u_1\otimes v_1\otimes w_1)+\cdots+(u_k\otimes v_k\otimes w_k)
\]
for certain $u_1,\dots,u_k\in U$, $v_1,\dots,v_k\in V$, and $w_1,\dots,w_k\in W$.

If a tensor $T$ is given in this form, the question is whether it can also
be written as a sum of fewer than $k$ tensors of the form $u\otimes v\otimes w$.
If this is not the case, then $k$ is called the rank of~$T$.

Using the relations quoted above, it is easy to turn a given tensor representation
with $k$ terms into one with $k+1$ terms. To do so, just pick one of its terms,
say $u\otimes v\otimes w$, write $u$ as $u'+u''$ for some vectors $u',u''\in U$,
and replace the term $u\otimes v\otimes w$ by the two terms
$u'\otimes v\otimes w$ and $u''\otimes v\otimes w$. This operation is called a
\emph{split.}

It is less obvious whether we can go in the other direction. If we are lucky
and the given tensor representation contains two terms $u_i\otimes v_i\otimes w_i$
and $u_j\otimes v_j\otimes w_j$ that agree in two positions, say $v_i=v_j$ and $w_i=w_j$,
then we can merge them into $(u_i+u_j)\otimes v_i\otimes w_i$. This operation is
called a \emph{reduction.} A reduction is a split applied backwards.

If we are not lucky and the given tensor representation does not allow for a reduction,
this does in general not mean that $k$ is the tensor rank. We have seen this in the
example in the introduction: even though the tensor representation
\[
a_0\otimes b_0\otimes c_0 + a_1\otimes b_0\otimes c_1 +
a_0\otimes b_1\otimes c_1 + a_1\otimes b_1\otimes c_2
\]
does not admit a reduction, the tensor can also be written as
\[
a_0\otimes b_0\otimes(c_0-c_1)+
(a_0+a_1)\otimes(b_0+b_1)\otimes c_1+
a_1\otimes b_1\otimes(c_2-c_1),
\]
so the rank is (at most)~$3$.

As we saw, this representation can be found with the help of an operation that is
applicable to any two terms $u_i\otimes v_i\otimes w_i$ and $u_j\otimes v_j\otimes w_j$
that agree in one position, say $u_i=u_j$. For every choice $\lambda\in K$, we have
\begin{alignat*}1
  & u_i\otimes v_i\otimes w_i + u_i\otimes v_j\otimes w_j\\
  &= u_i\otimes (v_i-\lambda v_j)\otimes w_i + u_i\otimes v_j\otimes(w_j+\lambda w_i),
\end{alignat*}
so we can replace the two terms in the first row by the two terms in the second.
This operation is called a \emph{flip.}

In summary, a split increases the number of terms by one, a reduction decreases the
number of terms by one, and a flip leaves the number of terms unchanged. In fact, a
flip can be decomposed into a split followed by a reduction:
\begin{alignat*}1
  & u_i\otimes v_i\otimes w_i + u_i\otimes v_j\otimes w_j\\
  &= u_i\otimes (v_i-\lambda v_j+\lambda v_j)\otimes w_i + u_i\otimes v_j\otimes w_j\\
  &\stackrel{\smash{\clap{\scriptsize \strut split}}}=\quad
  u_i\otimes (v_i-\lambda v_j)\otimes w_i
  + u_i\otimes (\lambda v_j) \otimes w_i + u_i\otimes v_j\otimes w_j\\
  &= u_i\otimes (v_i-\lambda v_j)\otimes w_i
  + u_i\otimes v_j \otimes (\lambda w_i) + u_i\otimes v_j\otimes w_j\\
  &\stackrel{\smash{\clap{\scriptsize\strut reduction}}}=\quad
  u_i\otimes (v_i-\lambda v_j)\otimes w_i
  + u_i\otimes v_j \otimes(w_j+\lambda w_i).  
\end{alignat*}
Note also that flips are reversible: if a tensor representation $T'$ can be
reached from a tensor representation $T$ by applying a flip, then $T$ can also
be reached from $T'$ via a flip.

In~\cite{kauers23f}, Kauers and Moosbauer introduced the flip graph of a tensor $T\in
U\otimes V\otimes W$.  Every tensor representation is a vertex in this graph,
and two vertices are connected if one can be reached from the other by either a
flip or a reduction. For some matrix multiplication tensors, they found shorter
representations than previously known by performing a random search in the flip
graph. Arai, Ichikawa and Hukushima~\cite{adaptive} enriched the flip graph with some
edges representing splits and found further improvements for some matrix formats
with their variant.

Having edges for splits in the graph has advantages and disadvantages. An
advantage is that in the version with split edges, the graph is strongly
connected, i.e., any tensor representation can be reached from any other via a
finite path. This was shown in Thm.~9 of~\cite{kauers23f} for the matrix
multiplication tensor but applies \emph{mutatis mutandis} to every tensor. A
disadvantage is that for every tensor representations there is an extremely
large number of options for applying a split. In a random search, it is not
clear how to make a reasonable choice.

Flips can be seen as an attempt to make reasonable choices for splits: only
choose such splits that allow for doing a reduction in the next step, because in
this combination, the number of terms in the tensor representation does not go
up. However, with only flips and reductions, the graph is no longer strongly
connected. See the discussion of Kauers and Moosbauer~\cite{kauers23f} for examples.

For matrix multiplication tensors, we do not know whether a representation of
minimal length can always be reached from the standard algorithm using only
flips and reductions. Here we will show that flips and reductions are sufficient
for polynomial multiplication tensors.

\section{Polynomial Multiplication}

For a field $K$, we consider the multiplication of a polynomial of degree at most $n$
with a polynomial of degree at most~$m$.
For $k\in\set N$, we write $K[x]_{\leq k}$ for the $K$-vector space of all polynomials
of degree at most~$k$. A basis of this vector space is $\{1,x,\dots,x^k\}$.
We are dealing with three such spaces: $K[x]_{\leq n}$, $K[x]_{\leq m}$, and $K[x]_{\leq n+m}$,
and it will be helpful to give distinct names to the basis elements of these three
spaces.
The basis elements $1,x,\dots,x^n$ of $K[x]_{\leq n}$ will be called $a_0,\dots,a_n$,
those of $K[x]_{\leq m}$ will be called $b_0,\dots,b_m$,
and those of $K[x]_{\leq n+m}$ will be called $c_0,\dots,c_{n+m}$.

\begin{defi}
  The tensor 
  \[
  \sum_{i=0}^n\sum_{j=0}^m a_i\otimes b_j\otimes c_{i+j}
  \in K[x]_{\leq n}\otimes K[x]_{\leq m}\otimes K[x]_{\leq n+m}
  \]
  is called the polynomial multiplication tensor for degrees $n$ and $m$ over~$K$,
  and this sum is called its standard representation.
\end{defi}

Polynomial multiplication is well understood, and asymptotically fast algorithms
for polynomial multiplication are well known. They are a staple of computer algebra
and covered in every textbook on computer algebra (cf. e.g., Chapter~8 of~\cite{vzgathen99}).
Fast algorithms for polynomial multiplication are based on the principle of
evaluation and interpolation: if $p$ and $q$ are polynomials of respective degrees
$n$ and $m$ and $f=pq$ is their product, then $f$ is the unique polynomial whose
values at $n+m+1$ distinct points $x_0,\dots,x_{n+m}\in K$ are the products $p(x_i)q(x_i)$
of the values of $p$ and $q$ at these points.

According to Lagrange's interpolation formula, we have
\[
  f(x)=\sum_{k=0}^{n+m}p(x_k)q(x_k)\prod_{\ell\neq k}\frac{x-x_\ell}{x_k-x_\ell}.
\]
In order to express this using the notation $c_0,\dots,c_{n+m}$ for the basis elements of $K[x]_{\leq n+m}$,
let $\alpha_{\ell,k}$ be the coefficient of $x^\ell$ in $\prod_{\ell\neq k}\frac{x-x_\ell}{x_k-x_\ell}$,
so that 
\begin{alignat}1\label{eq:ck}
c^{(k)}&:=\prod_{\ell\neq k}\frac{x-x_\ell}{x_k-x_\ell} = \sum_{\ell=0}^{n+m}\alpha_{\ell,k}c_\ell.
\end{alignat}
We then have
\begin{alignat}1\label{eq:toomcook}
  &\sum_{k=0}^{n+m}\biggl(\bigl(\sum_{i=0}^n x_k^ia_i\bigr)\otimes\bigl(\sum_{j=0}^m x_k^jb_j\bigr)\otimes c^{(k)}\biggr)\\
  &=\sum_{k=0}^{n+m}\sum_{i=0}^n\sum_{j=0}^m\bigr( (x_k^i a_i)\otimes (x_k^j b_j)\otimes c^{(k)}\bigr)\notag\\
  &=\sum_{k=0}^{n+m}\sum_{i=0}^n\sum_{j=0}^m\bigr( a_i\otimes b_j\otimes (x_k^{i+j} c^{(k)})\bigr)\notag\\
  &=\sum_{i=0}^n\sum_{j=0}^m a_i\otimes b_j\otimes \sum_{k=0}^{n+m} x_k^{i+j} c^{(k)}.\notag
\end{alignat}
According to the following lemma, this is exactly the polynomial multiplication tensor.

\begin{lemma}\label{A}
  In the notation introduced above, we have
  \[
    \sum_{k=0}^{n+m} x_k^{i+j} c^{(k)} = c_{i+j}
  \]
  for all $i\in\{0,\dots,n\}$ and all $j\in\{0,\dots,m\}$.
\end{lemma}
\begin{proof}
  Consider the polynomials $p=a_i=x^i$ and $q=b_j=x^j$. Their product is $pq=x^{i+j}=c_{i+j}$.
  The values of $pq$ at $x_0,\dots,x_{n+m}$ are $x_0^{i+j},\dots,x_{n+m}^{i+j}$.
  Therefore, by the interpolation formula quoted above, we have
  \[
    c_{i+j}=pq=\sum_{k=0}^{n+m}p(x_k)q(x_k)c^{(k)}=\sum_{k=0}^{n+m}x_k^{i+j}c^{(k)},
  \]
  as claimed.
\end{proof}

The representation~\eqref{eq:toomcook} is the basis of the multiplication
algorithms of Toom and Cook. We therefore call it the Toom-Cook representation
of the polynomial multiplication tensor. It implies that the rank of the
polynomial multiplication tensor is at most $n+m+1$. Since
$c^{(0)},\dots,c^{(n+m)}$ are linearly independent over~$K$, the rank cannot be
smaller than $n+m+1$, so the Toom-Cook representation is optimal with
respect to the number of multiplications (see~\cite{bodrato07,zanoni10} for a
discussion of the required number of additions).

\section{There is a Path}

The purpose of this section is to prove the following result about the flip graph for polynomial multiplication.

\begin{thm}\label{mainthm}
In the flip graph for polynomial multiplication over a field containing at least $n+m+1$ distinct elements $x_0,\dots,x_{n+m}$, there
is a path consisting of at most $n m (2n+2m+1)$ flips and $n m$ reductions that leads from the standard representation 
\[
  \sum_{i=0}^n\sum_{j=0}^m\biggl( a_i\otimes b_j\otimes c_{i+j}\biggr)
\]
of rank $(n+1)(m+1)$ to the Toom-Cook representation
\[
  \sum_{k=0}^{n+m}\biggl( \bigl(\sum_{i=0}^nx_k^ia_i\bigr)\otimes\bigl(\sum_{j=0}^mx_k^jb_j\bigr)\otimes c^{(k)} \biggr)
\]
of rank $n+m+1$, with $c^{(k)}$ as introduced in~\eqref{eq:ck}.
\end{thm}

Since the proof of Thm.~\ref{mainthm} is somewhat technical, we first illustrate the construction
for the special case $n=m=1$. Start with the standard representation and use the expressions
from Lemma~\ref{A} for the $c_i$:
\begin{alignat*}1
     &a_0\otimes b_0\otimes(c^{(0)}+c^{(1)}+c^{(2)})\\
{}+{}&a_1\otimes b_0\otimes(x_0c^{(0)}+x_1c^{(1)}+x_2c^{(2)})\\
{}+{}&a_0\otimes b_1\otimes(x_0c^{(0)}+x_1c^{(1)}+x_2c^{(2)})\\
{}+{}&a_1\otimes b_1\otimes(x_0^2c^{(0)}+x_1^2c^{(1)}+x_2^2c^{(2)})
\end{alignat*}
We flip the first row with the second and the third with the fourth to the effect of eliminating
$c^{(2)}$ from the second and the fourth rows. This gives
\begin{alignat*}1
  &(a_0+x_2a_1)\otimes b_0\otimes(c^{(0)}+c^{(1)}+c^{(2)})\\
{}+{}&a_1\otimes b_0\otimes((x_0-x_2)c^{(0)}+(x_1-x_2)c^{(1)})\\
{}+{}&(a_0+x_2a_1)\otimes b_1\otimes(x_0c^{(0)}+x_1c^{(1)}+x_2c^{(2)})\\
{}+{}&a_1\otimes b_1\otimes(x_0(x_0-x_2)c^{(0)}+x_1(x_1-x_2)c^{(2)}).
\end{alignat*}
Now we flip the first row with the third and the second with the fourth to the effect of eliminating
$c^{(1)}$ from the second and the fourth rows. This gives
\begin{alignat*}1
 &(a_0+x_2a_1)\otimes(b_0+x_1b_1)\otimes(c^{(0)}+c^{(1)}+c^{(2)})\\
{}+{}&a_1\otimes(b_0+x_1b_1)\otimes((x_0-x_2)c^{(0)}+(x_1-x_2)c^{(1)})\\
{}+{}&(a_0+x_2a_1)\otimes b_1\otimes((x_0-x_1)c^{(0)}+(x_2-x_1)c^{(2)})\\
{}+{}&a_1\otimes b_1\otimes((x_0-x_1)(x_0-x_2)c^{(0)}),
\end{alignat*}
where the fourth row can also be written as
\[
  ((x_0-x_1)a_1)\otimes((x_0-x_2)b_1)\otimes c^{(0)}.
\]
Flip this row with the third to the effect of eliminating $c^{(0)}$ from the third row. This gives
\begin{alignat*}1
 &(a_0+x_2a_1)\otimes(b_0+x_1b_1)\otimes(c^{(0)}+c^{(1)}+c^{(2)})\\
{}+{}&a_1\otimes(b_0+x_1b_1)\otimes((x_0-x_2)c^{(0)}+(x_1-x_2)c^{(1)})\\
{}+{}&(a_0+x_2a_1)\otimes b_1\otimes((x_2-x_1)c^{(2)})\\
{}+{}&(a_0+x_0a_1)\otimes b_1\otimes((x_0-x_1)c^{(0)}),
\end{alignat*}
where the third and the fourth rows can also be written as
\begin{alignat*}1
{}+{}&(a_0+x_2a_1)\otimes((x_2-x_1)b_1)\otimes c^{(2)}\\
{}+{}&(a_0+x_0a_1)\otimes((x_0-x_1)b_1)\otimes c^{(0)}.
\end{alignat*}
Flip the third row with the first to the effect of eliminating $c^{(2)}$ from the first row. This gives
\begin{alignat*}1
 &(a_0+x_2a_1)\otimes(b_0+x_1b_1)\otimes(c^{(0)}+c^{(1)})\\
{}+{}&a_1\otimes(b_0+x_1b_1)\otimes((x_0-x_2)c^{(0)}+(x_1-x_2)c^{(1)})\\
{}+{}&(a_0+x_2a_1)\otimes(b_0+x_2b_1)\otimes c^{(2)}\\
{}+{}&(a_0+x_0a_1)\otimes((x_0-x_1)b_1)\otimes c^{(0)}.
\end{alignat*}
Rewrite the second row into
\[
 ((x_1-x_2)a_1)\otimes(b_0+x_1b_1)\otimes(\frac{x_0-x_2}{x_1-x_2}c^{(0)}+c^{(1)})
\]
and flip this row with the first to the effect of eliminating $c^{(1)}$ from the second row.
This gives
\begin{alignat*}1
&(a_0+x_1a_1)\otimes(b_0+x_1b_1)\otimes(c^{(0)}+c^{(1)})\\
{}+{}&((x_1-x_2)a_1)\otimes(b_0+x_1b_1)\otimes(\frac{x_0-x_1}{x_1-x_2}c^{(0)})\\
{}+{}&(a_0+x_2a_1)\otimes(b_0+x_2b_1)\otimes c^{(2)}\\
{}+{}&(a_0+x_0a_1)\otimes ((x_0-x_1)b_1)\otimes c^{(0)}.
\end{alignat*}
Write the second row again as
\[
((x_0-x_1)a_1)\otimes(b_0+x_1b_1)\otimes c^{(0)}
\]
and finally flip this row with the first to the effect of eliminating $c^{(0)}$ from the first row.
This gives
\begin{alignat*}1
 &(a_0+x_1a_1)\otimes(b_0+x_1b_1)\otimes c^{(1)}\\
{}+{}&(a_0+x_0a_1)\otimes(b_0+x_1b_1)\otimes c^{(0)}\\
{}+{}&(a_0+x_2a_1)\otimes(b_0+x_2b_1)\otimes c^{(2)}\\
{}+{}&(a_0+x_0a_1)\otimes((x_0-x_1)b_1)\otimes c^{(0)}.
\end{alignat*}
Now a reduction is applicable to the second and the fourth row.

Although this derivation was somewhat longer than the one presented in the introduction,
it is worth pointing out that for general $n$ and~$m$, the path length announced in
Thm.~\ref{mainthm} is short compared to the length of a path involving splits constructed
as in Thm.~9 of~\cite{kauers23f}. The construction of this path is most easily described
backwards: starting from the Toom-Cook representation, multiply out all terms using
splits, and then merge common terms using reductions.
In the Toom-Cook representation for polynomials of degrees $n$ and~$m$, there are
$n+m+1$ tensors of rank one, and in each of them, the first factor is a sum of $n+1$
terms, the second factor is a sum of $m+1$ terms, and the third factor is a sum
of $n+m+1$ terms. Multiplying all of them out requires $nm(n+m)(n+m+1)^2$ splits.
To get from here to the standard representation, which consists of $(n+1)(m+1)$ terms,
we need $nm(n+m)(n+m+1)^2 - (n+1)(m+1)$ reductions.
Altogether, the length of the resulting path quartic while the length of the path
of Thm.~\ref{mainthm} is only cubic.

We now turn to the proof of Thm.~\ref{mainthm}. To prepare for it, we need the following
two lemmas.

\begin{lemma}\label{B}
It takes no more than $nm+(n-1)(m-1)+n$ flips to get from
\[
 \sum_{i=0}^n\sum_{j=0}^m\biggl( a_i\otimes b_j\otimes\underbrace{\sum_{\ell=0}^{n+m}x_\ell^{i+j}c^{(\ell)}}_{=c_{i+j}}\biggr)
\]
to
\begin{alignat*}3
  \bigl(\sum_{i=0}^nx_{n+m}^ia_i\bigr)&\otimes\bigl(\sum_{k=0}^mx_{n+m}^kb_k\bigr)&&\otimes\sum_{\ell=0}^{n+m}c^{(\ell)}\\
  +{}\sum_{i=1}^n\biggl( a_i&\otimes\bigl(\sum_{k=0}^mx_{n+m}^kb_k\bigr)&&\otimes\sum_{\ell=0}^{n+m-1}(x_\ell^i-x_{n+m}^i)c^{(\ell)} \biggr)\\
  +{}\sum_{i=0}^n\sum_{j=1}^m\biggl(a_i&\otimes\bigl(\sum_{k=j}^mx_{n+m}^{k-j}b_k\bigr)&&\otimes\sum_{\ell=0}^{n+m-1}x_\ell^{i+j-1}(x_\ell-x_{n+m})c^{(\ell)} \biggr)
\end{alignat*}
\end{lemma}
\begin{proof}
In order to be able to refer to terms in the tensor representations, let's number them.
We shall refer to the term $a_i\otimes b_j\otimes\sum_{\ell=0}^{n+m}x_\ell^{i+j}c^{(\ell)}$
as the $(i,j)$th term. While the terms get modified by flips, we maintain the numbering.

We apply several groups of flips.

First, for each $i$ from $1$ to $n$ and for each $j$ from $1$ to~$m$, apply a flip
to the $(i,0)$th and the $(i,j)$th term
in order to eliminate the term $x_{n+m}^{i+j}c^{(n+m)}$ from the third factor of the $(i,j)$th term.
This changes the second factors of the terms $(i,0)$ from $b_0$ to $b_0+b_1x_{n+m}+\cdots+b_mx_{n+m}^m$
and the third factors of the terms $(i,j)$ for $j>0$ from $\sum_{\ell=0}^{n+m}x_\ell^{i+j}c^{(\ell)}$
to $\sum_{\ell=0}^{n+m-1}(x_\ell^{i+j}-x_{n+m}^{i+j})c^{(\ell)}$.

Second, for each $j$ from $1$ to $m$ and each $i$ from $n$ down to~$1$, apply a flip
to the $(i,j)$th and the $(i-1,j)$th term in order to turn the third factor
$\sum_{\ell=0}^{n+m-1}(x_\ell^{i+j}-x_{n+m}^{i+j})c^{(\ell)}$ of the $(i,j)$th
term to $\sum_{\ell=0}^{n+m-1}x_\ell^{i+j-1}(x_\ell-x_{n+m})c^{(\ell)}$.
This has the side effect that $x_{m+n}$-fold of the second factor of the $(i,j)$th
term is added to the second factor of the $(i-1,j)$th term.
After all these flips have been performed, the second factor of the $(i,j)$th
term is $b_j+x_{n+m}b_{j+1}+\cdots+x_{n+m}^{m-j}b_m$.

Third, for each $i$ from $1$ to~$n$, apply a flip to the $(0,0)$th and the $(i,0)$th term
in order to change the third factor of the $(i,0)$th term from
$\sum_{\ell=0}^{n+m}x_\ell^i c^{(\ell)}$ to $\sum_{\ell=0}^{n+m-1}(x_\ell^i-x_{n+m}^i)c^{(\ell)}$.
All these flips change the first factor of the $(0,0)$th term to $a_0+x_{n+m}a_1+\cdots+x_{n+m}^na_n$.

At this point, we have reached the announced tensor representation.
We applied $nm$ flips in the first phase, $(n-1)(m-1)$ flips in the second phase,
and $n$ flips in the third phase. 
\end{proof}

\begin{lemma}\label{C}
It takes no more than $2n+1$ flips to get from
\begin{alignat*}1
  &\bigl(\sum_{i=0}^nx^ia_i\bigr)\otimes\bigl(\sum_{j=0}^mx^jb_j\bigr)\otimes(c_y + T_0)\\
  +{}&\sum_{i=1}^n\biggl( (y^i-x^i)a_i\otimes\bigl(\sum_{j=0}^mx^jb_j\bigr)\otimes(c_y + T_i)\biggr)\\
  +{}&\bigl(\sum_{i=0}^ny^ia_i\bigr)\otimes\sum_{i=1}^m(y^i-x^i)b_i\otimes c_y
\end{alignat*}
to
\begin{alignat*}1
  &\bigl(\sum_{i=0}^nx^ia_i\bigr)\otimes\bigl(\sum_{j=0}^mx^jb_j\bigr)\otimes T_0\\
  +{}&\sum_{i=1}^n\biggl((y^i-x^i)a_i\otimes(\sum_{j=0}^mx^jb_j)\otimes T_i\biggr)\\
  +{}&\bigl(\sum_{i=0}^ny^ia_i\bigr)\otimes\bigl(\sum_{j=0}^my^jb_j\bigr)\otimes c_y.
\end{alignat*}
\end{lemma}
\begin{proof}
Number the terms in the tensor representation from $0$ to $n+1$, so that, in particular,
the $i$th term in this numbering is also the $i$th of the sum in the second line.

First apply a flip to the $i$th and the $1$st term, for $i=2,\dots,n$, so that the
first factor of the $1$st term becomes $\sum_{i=1}^n (y^i-x^i)a_i$ and the third
factor of the $i$th term ($i\geq2$) becomes $T_i-T_1$. This takes $n-1$ flips.

Next, apply a flip to the $0$th and the $1$st term, so that the first factor of the
$1$st term becomes $\sum_{i=0}^n y^ia_i$ and the third factor of the $0$th term
becomes $T_0-T_1$. This takes $1$ flip.

Next, apply a flip to the $1$st term and the $(n+1)$st term, so that the second
factor of the $n+1$st term becomes $\sum_{j=0}^my^jb_j$ and the third factor of
the $1$st term becomes $-c_y+T_1$. This takes $1$ flip.

Next, apply a flip to the $0$th and the $1$st term, so that the third factor of the
$0$st term becomes $T_0$ and the first factor of the $1$st term becomes $\sum_{i=1}^n (y^i-x^i)a_i$.
This takes $1$ flip.

Finally, apply a flip to the $i$th and the $1$st term, for $i=2,\dots,n$, so that
the first factor of the $1$st term becomes $(y-x)a_1$ and the third factor of the $i$th
term ($i\geq2$) becomes~$T_i$. This takes $n-1$ flips.

At this point, we have reached the announced tensor representation.
Altogether, we have made $(n-1)+1+1+1+(n-1)=2n+1$ flips.
\end{proof}

\begin{proof}[Proof of Thm.~\ref{mainthm}]
We apply induction on $n+m$. The case $n=m=0$ serves as induction base. In this case,
the tensor in question is just $a_0\otimes b_0\otimes c_0$. This is the tensor representation
in the beginning and in the end, so there obviously is a path connecting them, and it
consists of zero flips.

Now let $n,m\geq0$ be given and assume that the theorem is true for all $n',m'$ with
$n'+m'<n+m$. Because of symmetry, we may assume that $m\geq n$. Note that this implies
$m\geq1$.

We first apply the flips of Lemma~\ref{B}. After introducing new variables
\begin{alignat*}3
  \tilde b_j &= b_{j+1} + x_{n+m}b_{j+2} + \cdots + x_{n+m}^{m-j}b_m&\quad&(j=0,\dots,m-1)\\
  \tilde c^{(\ell)} &= (x_\ell - x_{n+m})c^{(\ell)}&&(\ell=0,\dots,n+m-1),
\end{alignat*}
the resuting tensor representation has the following form:
\begin{alignat}1
  &\bigl(\sum_{i=1}^nx_{n+m}^ia_i\bigr)\otimes\bigl(\sum_{j=0}^mx_{n+m}^jb_j\bigr)\otimes\sum_{\ell=0}^{n+m}c^{(\ell)}\notag\\
  +{}&\sum_{i=1}^n\biggl( a_i\otimes\bigl(\sum_{j=0}^mx_{n+m}^jb_j\bigr)\otimes\sum_{\ell=0}^{n+m-1}(x_\ell^i-x_{n+m}^i)c^{(\ell)} \biggr)\label{eq:1}\\
  +{}&\sum_{i=0}^n\sum_{j=0}^{m-1}\biggl(a_i\otimes \tilde b_j\otimes\sum_{\ell=0}^{n+m-1}x_\ell^{i+j-1}\tilde c^{(\ell)} \biggr).\notag
\end{alignat}
By induction hypothesis, the double sum in the third line of~\eqref{eq:1} can be turned into
\[
\sum_{\ell=0}^{n+m-1}\biggl(\bigl(\sum_{i=0}^nx_\ell^ia_i\bigr)\otimes\bigl(\sum_{j=0}^{m-1}x_\ell^j\tilde b_j\bigr)
\otimes\tilde c^{(\ell)}\biggr).
\]
Undoing the change of variables yields
\begin{alignat*}1
  &\sum_{\ell=0}^{n+m-1}\biggl(\bigl(\sum_{i=0}^nx_\ell^ia_i\bigr)
  \otimes\sum_{j=0}^{m-1}\sum_{i=0}^j x_\ell^jx_{n+m}^{i-j}b_{j+1}
  \otimes(x_\ell-x_{n+m}) c^{(\ell)}\biggr)\\
  &=\sum_{\ell=0}^{n+m-1}\biggl(\bigl(\sum_{i=0}^nx_\ell^ia_i\bigr)
  \otimes\sum_{j=1}^m(x_\ell^j - x_{m+n}^j)b_j
  \otimes c^{(\ell)}\biggr).
\end{alignat*}
We replace the double sum in the third row of \eqref{eq:1} by this expression, and in each summand
of the sum in the second row, we move a factor of $(x_0^i-x_{n+m}^i)$ from the third factor to the first.
This transforms the tensor representation of \eqref{eq:1} into
\begin{alignat*}1
  &\bigl(\sum_{i=0}^nx_{n+m}^ia_i\bigr)\otimes\bigl(\sum_{j=0}^mx_{n+m}^jb_j\bigr)\otimes\bigl(c^{(0)} + \sum_{\ell=1}^{n+m}c^{(\ell)}\bigr)\\
  +{}&\sum_{i=1}^n\biggl( (x_0^i-x_{n+m}^i)a_i\otimes\bigl(\sum_{j=0}^mx_{n+m}^jb_j\bigr)\otimes\bigl(c^{(0)}+ \sum_{\ell=1}^{n+m}\frac{x_\ell^i-x_{n+m}^i}{x_0^i-x_{n+m}^i}c^{(\ell)}\bigr)\biggr)\kern-7pt\null\\
  +{}&\sum_{\ell=0}^{n+m-1}\biggl(\bigl(\sum_{i=0}^nx_\ell^ia_i\bigr)
  \otimes\sum_{j=1}^m(x_\ell^j - x_{m+n}^j)b_j
  \otimes c^{(\ell)} \biggr).
\end{alignat*}
Now we are in a position to apply Lemma~\ref{C}, with $x_{n+m}$, $x_0$, $c^{(0)}$, $\sum_{\ell=1}^{n+m}c^{(\ell)}$,
and $\sum_{\ell=1}^{n+m}\frac{x_\ell^i-x_{n+m}^i}{x_0^i-x_{n+m}^i}c^{(\ell)}$ in the roles
of $x, y, c_y, T_0, T_1$, respectively. At the cost of $2n+1$ flips, we arrive at
\begin{alignat*}1
  &\bigl(\sum_{i=0}^nx_{n+m}^ia_i\bigr)\otimes\bigl(\sum_{j=0}^mx_{n+m}^jb_j\bigr)\otimes\bigl(\sum_{\ell=1}^{n+m}c^{(\ell)}\bigr)\\
  +{}&\sum_{i=1}^n\biggl( (x_0^i-x_{n+m}^i)a_i\otimes\bigl(\sum_{j=0}^mx_{n+m}^jb_j\bigr)\otimes\bigl(\sum_{\ell=1}^{n+m-1}\frac{x_\ell^i-x_{n+m}^i}{x_0^i-x_{n+m}^i}c^{(\ell)}\bigr)\biggr)\\
  +{}&\bigl(\sum_{i=0}^nx_0^ia_i\bigr)\otimes\bigl(\sum_{j=0}^mx_0^jb_j\bigr)\otimes c^{(0)}\\
  +{}&\sum_{\ell=1}^{n+m-1}\biggl(\bigl(\sum_{i=0}^nx_\ell^ia_i\bigr)
  \otimes\sum_{j=1}^{m}(x_\ell^j - x_{m+n}^j)b_j
  \otimes c^{(\ell)} \biggr).
\end{alignat*}
Next, we move a factor of $(x_1^i-x_{n+m}^i)$ from the third factor to the first factor in each
summand of the sum in the second row, so that we can apply Lemma~\ref{C} again. After altogether
$n+m-1$ applications of Lemma~\ref{C}, we arrive at
\begin{alignat*}1
  &\bigl(\sum_{i=0}^nx_{n+m}^ia_i\bigr)\otimes\bigl(\sum_{j=0}^mx_{n+m}^jb_j\bigr)\otimes(c^{(n+m-1)} + c^{(n+m)})\\
  +{}&\sum_{i=1}^n\biggl( (x_{n+m-1}^i-x_{n+m}^i)a_i\otimes\bigl(\sum_{j=0}^mb_jx_{n+m}^j\bigr)\otimes c^{(n+m-1)} \biggr)\\
  +{}&\sum_{\ell=0}^{n+m-2}\biggl(\bigl(\sum_{i=0}^nx_\ell^ia_i\bigr)\otimes\bigl(\sum_{j=0}^mx_\ell^jb_j\bigr)\otimes c^{(\ell)}\biggr)\\
  +{}&\bigl(\sum_{i=0}^na_ix_{n+m-1}^i\bigr)\otimes\sum_{j=1}^m(x_{n+m-1}^j - x_{m+n}^j)b_j \otimes c^{(n+m-1)}.
\end{alignat*}
Using $n-1$ reductions, the sum in the second line can be merged into a single term
\[
 \sum_{i=1}^n(x_{n+m-1}^i-x_{n+m}^i)a_i\otimes\bigl(\sum_{j=0}^mx_{n+m}^jb_j\bigr)\otimes c^{(n+m-1)},
\]
where the summation sign is now understood as part of the first factor. Flip this line with the first line
to obtain
\begin{alignat*}1
  &\bigl(\sum_{i=0}^nx_{n+m}^ia_i\bigr)\otimes\bigl(\sum_{j=0}^mx_{n+m}^jb_j\bigr)\otimes c^{(n+m)}\\
  +{}&\bigl(\sum_{i=0}^nx_{n+m-1}^ia_i\bigr)\otimes\bigl(\sum_{j=0}^mx_{n+m}^jb_j\bigr)\otimes c^{(n+m-1)}\\
  +{}&\sum_{\ell=0}^{n+m-2}\biggl((\sum_{i=0}^nx_\ell^ia_i)\otimes(\sum_{j=0}^mx_\ell^jb_j)\otimes c^{(\ell)}\biggr)\\
  +{}&\bigl(\sum_{i=0}^nx_{n+m-1}^ia_i\bigr)\otimes\bigl(\sum_{j=1}^m(x_{n+m-1}^j - x_{m+n}^j)b_j\bigr)\otimes c^{(n+m-1)}.
\end{alignat*}
Finally, we apply a reduction to the second and the last row to arrive at the desired tensor representation.

If $F(n,m)$ denotes the total number of flips we used, then we have
\begin{alignat*}1
  F(n,m) &= \overbrace{nm+(n-1)(m-1)+n}^{\text{Lemma~\ref{B}}} + \overbrace{F(n,m-1)}^{\text{Ind.Hyp.}}\\
  &\qquad{}+(n+m-1)\times\underbrace{(2n+1)}_{\text{Lemma~\ref{C}}} + 1.
\end{alignat*}
It can be easily checked that $m n (2n+2m+1)$ matches this recurrence and the initial values, so
we have $F(n,m)=mn(2n+2m+1)$ for all $n,m$.

It is also clear that there are altogether $nm$ reductions, because we start from a tensor of length
$(n+1)(m+1)$ and we finish with a tensor of length $n+m-1$, and
$nm=(n+1)(m+1)-(n+m-1)$. This completes the proof.
\end{proof}

\section{Small Fields}\label{sec:chartwo}

Unlike Karatsuba's algorithm, which works for every ground domain, the Toom-Cook representations
of the polynomial multiplication tensor only exist when the ground field is sufficiently large:
Thm.~\ref{mainthm} requires $K$ to contain at least $n+m+1$ distinct elements that can serve as base
points for the evaluation and interpolation.

What happens if $K$ is smaller, for example, for $K=\set Z_2$? We have adapted the software used
by Kauers and Moosbauer~\cite{kauers23f} to the polynomial case and used it to search for low rank
representations of the polynomial multiplication tensor for this case.
We also included some of the improvements proposed in~\cite{adaptive}.
The results are given in the following table. For example, for the degrees $n=3$ and $m=4$
we found a representation of rank~$12$.

\begin{center}
\def\entry#1#2#3#4{\draw(#1,#2)node{#3\rlap{$^{#4}$}};}
\begin{tikzpicture}[yscale=-.5,scale=.67]
\foreach\i in{1,2,3,4,5,6,7,8,9,10} {\draw(\i,0)node{\textit{\i}};}
\foreach\i in{1,2,3,4,5,6,7,8,9,10} {\draw(11,\i)node{\textit{\i}};}
\entry{1}{1}{3}{*}
\entry{2}{1}{5}{*}
\entry{2}{2}{6}{*}
\entry{3}{1}{6}{*}
\entry{3}{2}{8}{*}
\entry{3}{3}{9}{*}
\entry{4}{1}{8}{*}
\entry{4}{2}{10}{*}
\entry{4}{3}{12}{*}
\entry{4}{4}{13}{*}
\entry{5}{1}{9}{*}
\entry{5}{2}{11}{*}
\entry{5}{3}{13}{*}
\entry{5}{4}{16}{*}
\entry{5}{5}{17}{+}
\entry{6}{1}{11}{*}
\entry{6}{2}{13}{*}
\entry{6}{3}{15}{*}
\entry{6}{4}{18}{*}
\entry{6}{5}{20}{*}
\entry{6}{6}{22}{+}
\entry{7}{1}{12}{*}
\entry{7}{2}{15}{*}
\entry{7}{3}{17}{*}
\entry{7}{4}{19}{*}
\entry{7}{5}{22}{+}
\entry{7}{6}{24}{\cdot}
\entry{7}{7}{26}{\cdot}
\entry{8}{1}{14}{*}
\entry{8}{2}{17}{+}
\entry{8}{3}{19}{+}
\entry{8}{4}{22}{+}
\entry{8}{5}{24}{+}
\entry{8}{6}{26}{\cdot}
\entry{8}{7}{28}{\cdot}
\entry{8}{8}{32}{}
\entry{9}{1}{15}{*}
\entry{9}{2}{19}{\cdot}
\entry{9}{3}{21}{*}
\entry{9}{4}{22}{}
\entry{9}{5}{26}{\cdot}
\entry{9}{6}{28}{}
\entry{9}{7}{32}{}
\entry{9}{8}{36}{}
\entry{9}{9}{39}{}
\entry{10}{1}{17}{*}
\entry{10}{2}{21}{\cdot}
\entry{10}{3}{23}{\cdot}
\entry{10}{4}{25}{\cdot}
\entry{10}{5}{28}{\cdot}
\entry{10}{6}{30}{}
\entry{10}{7}{37}{}
\entry{10}{8}{52}{}
\entry{10}{9}{64}{}
\entry{10}{10}{66}{}
\end{tikzpicture}
\end{center}

We tried to extend the schemes we found for the field $\set Z_2$ using Hensel lifting to schemes with
coefficients in~$\set Z_{2^{20}}$. Then we applied rational reconstruction to obtain schemes with coefficients
in $\set Q$, or, ideally, with coefficients in~$\set Z$. Schemes with coefficients in $\set Z$ are of
particular interest because they apply to every ground ring, in particular to other small fields
like $\set Z_3$ or~$\set Z_5$. Schemes with coefficients in $\set Q$ only apply to other fields whose
characteristic does not divide the denominator of any coefficient.

For the schemes marked with a star, the extension to integer coefficients was successful.
For the schemes marked with a plus, rational reconstruction led to coefficients in $\set Z[\frac1{105}]$.
This means that these schemes apply to fields with characteristic other than $3,5,7$.
For the schemes marked with a dot, we were able to lift the coefficients to $\set Z_{2^{20}}$, but
not to~$\set Q$. The schemes without any decoration did not admit Hensel lifting.

The first two rows of the table are easily explained. 

\begin{thm}
  \begin{enumerate}
  \item 
  In the flip graph for polynomial multiplication over an arbitrary field there
  is a path that leads from the standard representation for degrees $n$ and $1$
  to a representation of rank~$\lceil\frac32(n+1)\rceil$.
  \item
  In the flip graph for polynomial multiplication over an arbitrary field there 
  is a path that leads from the standard representation for degrees $n\geq5$ and $2$
  to a representation of rank~$2n+1$.
  \end{enumerate}
\end{thm}
\begin{proof}
  We proceed by induction on~$n$.
  \begin{enumerate}
  \item 
    For small~$n$, the claim is confirmed by the computation.
    The large~$n$, the claim follows from the observations that increasing $n$ by $1$ raises the rank by no more than $2$
    and increasing $n$ by $2$ raises the rank by no more than~$3$.

    The first observation simply follows from
    \begin{alignat*}1
    &\sum_{i=0}^{n+1}\sum_{j=0}^1 a_i\otimes b_j\otimes c_{i+j}\\
    &=\sum_{i=0}^n\sum_{j=0}^1 a_i\otimes b_j\otimes c_{i+j}\\
    &+a_{n+1}\otimes b_0\otimes c_{n+1}+a_{n+1}\otimes b_1\otimes c_{n+2},
    \end{alignat*}
    and the second observation follows from
    \begin{alignat*}1
    &\sum_{i=0}^{n+2}\sum_{j=0}^1 a_i\otimes b_j\otimes c_{i+j}\\
    &=\sum_{i=0}^n\sum_{j=0}^1 a_i\otimes b_j\otimes c_{i+j}\\
    &\quad+a_{n+1}\otimes b_0\otimes c_{n+1}
    +a_{n+1}\otimes b_1\otimes c_{n+2}\\
    &\quad+a_{n+2}\otimes b_0\otimes c_{n+2}
    +a_{n+2}\otimes b_1\otimes c_{n+3}\\
    &=\sum_{i=0}^n\sum_{j=0}^1 a_i\otimes b_j\otimes c_{i+j}\\
    &\quad + a_{n+1}\otimes b_0\otimes (c_{n+1}-c_{n+2})\\
    &\quad + (a_{n+1}+a_{n+2})\otimes(b_0+b_1)\otimes c_{n+2}\\
    &\quad + a_{n+2}\otimes b_1\otimes (c_{n+3}-c_{n+2}).
    \end{alignat*}
  \item
    For small~$n$, the claim is confirmed by the computation.
    The large~$n$, the claim follows from the observation that increasing $n$ by $3$ raises the rank by no more than~$6$.
    This follows from
    \begin{alignat*}1
      &\sum_{i=0}^{n+3}\sum_{j=0}^2 a_i\otimes b_j\otimes c_{i+j}\\
      &=\sum_{i=0}^{n}\sum_{j=0}^2 a_i\otimes b_j\otimes c_{i+j}
    +\sum_{i=n+1}^{n+3}\sum_{j=0}^2 a_i\otimes b_j\otimes c_{i+j}
    \end{alignat*}
    and the fact that the multiplication of two quadratic polynomials has rank $6$ according to the computation.
  \end{enumerate}
\end{proof}

Flip graphs are useful for finding low-rank tensor representations, but it is not clear how to use
the technique for checking whether an optimum has been reached. It could always be that we fail to
find a representation with a lower rank because either the better scheme is not reachable without
splits or they are too well hidden in the graph. For example, according to the
table above, for $K=\set Z_2$ and $n=m=2$ we only find a representation of rank~$6$ using flip
graphs, while for larger fields, there is a representation of rank~$5$. Is the result found by the
flip graph search optimal for $K=\set Z_2$, or did we miss a better scheme?

In order to answer this question, we can search for a scheme of rank $5$ by solving a nonlinear
system. Consider a representation with undetermined coefficients,
\[
\sum_{\ell=1}^5 \biggl(\sum_{i=0}^2 \alpha_{\ell,i}a_0\biggr)\otimes
\biggl(\sum_{j=0}^2 \beta_{\ell,j}b_j\biggr)\otimes
\biggl(\sum_{k=0}^4 \gamma_{\ell,k}c_k\biggr),
\]
equate it to the polynomial multiplication tensor
\[
  \sum_{i=0}^2\sum_{j=0}^2 a_i\otimes b_j\otimes c_{i+j},
\]
and compare coefficients. This leads to the system of equations
\[
\sum_{\ell=1}^5 \alpha_{\ell,i}\beta_{\ell,j}\gamma_{\ell,k} = \delta_{i+j,k}
\]
for $i,j=0,1,2$ and $k=0,\dots,4$, where $\delta$ refers to the Kronecker delta.
The analogous equations for the matrix multiplication tensor are known as the Brent equations~\cite{Br:Afmm}.
In~\cite{HKS:Nwtm}, Heule, Kauers, and Seidl found many new representations of the matrix multiplication
tensor for $3\times3$ matrices by translating the equations into a boolean formula and solving it
using a SAT solver. Here we used a SAT solver to prove that for $K=\set Z_2$ and $n=m=2$, there is no
representation of rank~$5$. Some further representations can be proven to be optimal in the same way:

\begin{thm}
  For $K=\set Z_2$, the flip graph for the polynomial multiplication tensor
  has a path from the standard representation to an representation of minimal
  rank for every
  \[
   (n,m)\in\{(1,1),(1,2),(1,3),(1,4),(1,5),(2,2),(2,3),(2,4),(3,3)\}.
  \]
\end{thm}
\begin{proof}
  For all these pairs $(n,m)$, we were able to prove with a SAT solver that there is no
  representation of smaller rank than what was found by the flip graph search as indicated
  in the table above. 
\end{proof}

We do not believe that all the rank bounds for higher degrees reported in the table are tight.
Like in the case of matrix multiplication~\cite{kauers23f,kauers24c,adaptive}, searching in
the flip graph becomes more and more cumbersome with increasing tensor size. 

\section{Conclusion}

We have seen that the concept of flip graphs also works well for the polynomial multiplication
tensor. If the ground field is sufficiently large, it is capable of finding, at least in principle,
an optimal representation starting from the standard representation in a relatively small number
of steps.

For fields with more than two elements, we also have to take into account that constant factors
can be freely moved between the components of a rank-one tensor: $(\alpha u)\otimes v\otimes w
=u\otimes(\alpha v)\otimes w=u\otimes v\otimes(\alpha w)$. We made free use of this relation in
the construction of the path. A search engine that follows a random path in the flip graph would
somehow have to cope with this freedom, and it is unclear what is the best way of doing this.
This may be an explanation why an automated search in the flip graph works best for $K=\set Z_2$.

For this case, we have found low-rank tensor representations for various small values of $n$ and~$m$.
For some of them, we were able to show that the representations are optimal, and hence that the
rank of the polynomial multiplication tensor over $\set Z_2$ is sometimes strictly larger than the
rank over larger fields. To get a clearer picture of the nature of the polynomial multiplication
tensor and of the structure of flip graphs, it would also be interesting to specifically search
for low rank representations for other small fields, e.g., $\set Z_3,\set Z_5$, or~$\set Z_7$,
in particular for those degrees where Hensel-lifting and rational reconstruction so far only led
to schemes with coefficients in $\set Z[\frac1{105}]$. 

Moreover, flip graphs could be used to analyze the tensor ranks for certain types of non-commutative
polynomial multiplication, i.e., the multiplication in a Weyl algebra. For large sizes, it is known
that the complexity of multiplication there is connected to the complexity of matrix
multiplication~\cite{bostan08b}. For small sizes, there may however be discrepancies.

Finally, the matrix multiplication tensor is of pivotal interest. For this tensor, it remains unclear
whether it is always possible to reach an optimal representation starting from the standard
representation (without splits). If so, we would like to know a bound on the length of the shortest
path and perhaps make some statements about its structure. Such understanding might be helpful for
the design of more efficient search techniques that could help to find low rank tensor representations
for matrix multiplication tensors of sizes that are currently out of reach.

\bibliographystyle{ACM-Reference-Format}
\balance
\bibliography{bib}

\end{document}